\documentclass[conference,letterpaper]{IEEEtran}

\usepackage[utf8]{inputenc} 
\usepackage[T1]{fontenc}
\usepackage{url}
\usepackage{ifthen}
\usepackage{cite}
\usepackage[cmex10]{amsmath} 

\usepackage{epic}
\usepackage{eepic}
\usepackage{color}
\usepackage{paralist}
\usepackage{mathtools}
\usepackage{amssymb}
\usepackage{graphicx}
\usepackage{caption}
\usepackage{subcaption}
\usepackage{amsthm}
\usepackage{wrapfig}
\usepackage{listings}
\usepackage{mdframed}
\usepackage{pgfplots}
\newtheorem{theorem}{Theorem}
\newtheorem*{theorem*}{Theorem}

\newtheorem{proposition}[theorem]{Proposition}
\newtheorem*{proposition*}{Proposition}
\newtheorem{lemma}[theorem]{Lemma}
\newtheorem*{lemma*}{Lemma}

\newtheorem*{conjecture*}{Conjecture}

\newtheorem*{fact*}{Fact}

\newtheorem*{hypothesis*}{Hypothesis}

\newtheorem*{claim*}{Claim}

\theoremstyle{definition}
\newtheorem{definition}[theorem]{Definition}
\newtheorem{construction}[theorem]{Construction}
\newtheorem{example}[theorem]{Example}

\newtheorem{algorithm}[theorem]{Algorithm}

\theoremstyle{remark}
\newtheorem{remark}[theorem]{Remark}
\newtheorem*{remark*}{Remark}

\newtheorem*{observation*}{Observation}

\newcommand{\bsc}{\mathsf{BSC}}

\newcommand{\ke}[1]{{\footnotesize\color{red}[Ke: #1]}}

\newcommand{\short}{1}

\pgfplotsset{compat=1.16}

\begin{document}
\title{A Practical Coding Scheme \\ for the BSC with Feedback} 



 \author{
   \IEEEauthorblockN{Ke Wu\IEEEauthorrefmark{1} and
                     Aaron B. Wagner\IEEEauthorrefmark{2}
                     }
   \IEEEauthorblockA{\IEEEauthorrefmark{1}%
                    Computer Science Department,
                    Carnegie Mellon University, Pittsburgh, PA 15213 USA.
                    kew2@andrew.cmu.edu.
                    }
   \IEEEauthorblockA{\IEEEauthorrefmark{2}%
                     School of Electrical and Computer Engineering,
                       Cornell University, Ithaca, NY 14850 USA.
                     wagner@cornell.edu.}
}

\maketitle


\begin{abstract}
We provide a practical implementation of the rubber method of Ahlswede~\emph{et al.} for binary alphabets. The idea is to create the ``skeleton'' sequence therein via an arithmetic decoder designed for a particular $k$-th order Markov chain. For the stochastic binary symmetric channel, we show that the scheme is nearly optimal in a strong sense for certain parameters.
\end{abstract}


\section{Introduction}

We consider the binary symmeric channel with ideal feedback, both in its stochastic- and adversarial-noise forms. In the former, each bit is flipped independently with some probability $p$. In the latter, an omniscient adversary can flip up to a fraction $f$ of the bits in order to disrupt the communication.

The information-theoretic limits for both forms of the channel, assuming perfect feedback, are well-known.  In the adversarial case, the capacity as a function of $f$ was determined by Zigangirov~\cite{zigangirov1976number}, building on earlier results of Berlekamp~\cite{berlekamp1968block}. For the stochastic version, the capacity equals that of the non-feedback version (e.g.,~\cite{cover2012elements, shannon1956zero}) and likewise the high-rate error exponent, normal approximation, and moderate deviations performance are all unimproved by feedback.  In fact, the third-order coding rate is unimproved by feedback~\cite{altuug2020exact}, as is the order of the optimal ``pre-factor'' in front of the error exponent at high rates.  Thus, at least for the stochastic version of the channel, feedback offers very little improvement in coding performance.

In general, feedback is known to simplify the coding problem even if it does not provide for improved performance. The erasure (e.g., \cite[Section 17.1]{el2011network}), and Gaussian channels \cite{schalkwijk1966coding1,schalkwijk1966coding2} provide striking examples of this phenomenon. For the BSC, see \cite{horstein1963sequential, schalkwijk1971class} for classical and\cite{9174445,9174232} for recent work on devising implementable schemes using feedback.

For the adversarial symmetric channel with feedback (and arbitrary, finite alphabet size), Ahlswede \emph{et al.}~\cite{4036418} proposed an explicit scheme called the \emph{rubber method}.  In the binary case, for a fixed $\ell > 2$, the message is encoded as a ``skeleton'' string containing no substring of $\ell$ consecutive zeros. The encoder then transmits this string, sending $\ell$ consecutive zeros to indicate that an error has occurred. For each $\ell$, this scheme achieves the capacity of the adversarial channel for a certain choice of $f$. This scheme simplifies significantly the original achievability argument of Berlekamp~\cite{berlekamp1968block}.  For ternary and larger alphabets, the scheme is even simpler. Rubber method has since been generalized~\cite{lebedev2016coding, deppe2020coding, deppe2020bounds}.

We only consider the binary case in this paper, and we make two contributions. The first is to propose the use of arithmetic coding applied to a particular Markov chain in order to efficiently encode the message sequence into the corresponding skeleton string. This results in a practically-implementable end-to-end scheme, with only a negligible rate penalty. The second contribution is showing that, for each $\ell$, there is a special rate $R_\ell^*$ and crossover probability $p_\ell$ such that the resulting scheme is optimal with respect to the second-order coding rate and moderate deviations performance for the channel with crossover probability $p_\ell$ and error-exponent optimal at rate $R_\ell^*$ for all channels with crossover probability less than $p_\ell$. We also consider the third-order coding rate and the ``pre-factor'' of the error exponent of the scheme. These turn out to be nearly, but not exactly optimal. See Section~\ref{sec:optimal}.

In Section \ref{sec:prelim} we introduce our notation and provide various preliminaries. In Section \ref{sec:markov} and \ref{sec:coding} and we describe our coding scheme. In Section \ref{sec:optimal} we present our main results.
\if0\short{The proofs are omitted due to space constraints, but are available in the full version: ?? ~\cite{}}\fi
\section{Notation and Preliminaries}
\label{sec:prelim}
Capital letters such as $X$ or $Y$ denote random variables. We use $x^{n}$ to denote the first $n$ bits of the sequence $x_1,\dots,x_n$,
and we use $z\|z'$ to denote the concatenation of two strings $z$ and $z'$. In addition, $\lfloor x^N\rfloor_{L}$ denotes the truncation of $x^N$ to the first $L$ bits.

We use ${\sf Bin}(n,p)$ to denote the binomial distribution with size $n$ and success probability $p$ and $\mathcal{N}(\mu,\sigma^2)$ to denote the normal distribution with mean $\mu$ and variance $\sigma^2$. Moreover, $B(p)$ denotes the Bernoulli distribution with success probability $p$. We use $D(P\|Q)$ to denote the Kullback-Leibler divergence between distribution $P$ and $Q$. 

\subsection{The Channel Model}
Let $\bsc(p)$ denote a binary symmetric channel with cross-over probability $p\in(0,\frac{1}{2})$ without feedback. That is, $\bsc(p)$ has input alphabet $\mathcal{X}=\{0,1\}$ and output alphabet $\mathcal{Y}=\{0,1\}$, and probability transition matrix
\begin{equation*}
p(y|x)=
    \begin{pmatrix}
    1-p& p\\
    p & 1-p\\
    \end{pmatrix}.
\end{equation*}

Suppose that an encoder wishes to send a message $m$ in a message space $\mathcal{M}$ through $\bsc(p)$. It first encodes the message $m$ using an encoding function $f$, and sends $x^N=f(m)$ through the channel. The decoder, upon receiving $y^N$ from the channel, runs a decoding function $g$ on $y^N$ to obtain $m'$. The pair $(f,g)$ is called a \emph{code} $\mathcal{C}_{N,R}$ with block length $N$ and rate $R=\frac{\log |\mathcal{M}|}{N}$. The (average) error probability of a code $\mathcal{C}_{N,R}$ is defined as $P_e(\mathcal{C}_{N,R}):=\frac{1}{|\mathcal{M}|}\sum_{m\in\mathcal{M}}\Pr[m'\neq m]$. 

The \emph{capacity} of $\bsc(p)$ is well-known to be
\[C(\bsc(p))=1-h(p),\]
where $h(\cdot)=-p\log p-(1-p)\log (1-p)$ is the binary entropy function, and the $\log$ is base-$2$ throughout.

We will also consider the adversarial binary symmetric channel $\bsc_{adv}(f)$ in which at most $f$ fraction of transmitted bits can be adversarially flipped.

Feedback allows the encoder to see exactly what the decoder receives after each transmission and update its next transmission accordingly. In the BSC with feedback, which we denote as $\bsc^{fb}(p)$, the encoding function $f$ consists of a sequence of maps $\{f_i\}_{i=1}^N$. Each $f_i$ takes as input $m,y_1,\dots,y_{i-1}$, and outputs $x_i$, the next bit to send. The decoder then runs $g(y^N)$ to obtain $m'$.

It is well-known that feedback does not improve the channel capacity: \[C(\bsc^{fb}(p))=1-h(p).\] 
For the adversarial feedback BSC channel $\bsc_{adv}^{fb}(f)$, an upper bound on the capacity was first shown by Berlekamp \cite{berlekamp1968block}. He also gives a lower bound that coincides with the upper bound when $f\geq \frac{3-\sqrt{5}}{4}$. A lower bound that coincides with the upper bound for $f<\frac{3-\sqrt{5}}{4}$ was given by Zigangirov \cite{zigangirov1976number}, thus determining the capacity for $\bsc_{adv}^{fb}(f)$:
\begin{equation*}
    C(\bsc_{adv}^{fb}(f))=
    \begin{cases}
    1-h(f) &\text{ if }0\leq f\leq \frac{3-\sqrt{5}}{4},\\
    (1-3f)\log\frac{1+\sqrt{5}}{2} &\text{ if }\frac{3-\sqrt{5}}{4}<f\leq 1.
    \end{cases}
\end{equation*}

\begin{figure}
\centering
\begin{tikzpicture}[trim axis left]
\begin{axis}[grid=major, xmin=0, xmax=0.5, ymin=0, ymax=1,
     xlabel=$p$, ylabel={$R$},
     xtick = {0,0.1910, 1/3, 0.5}, xticklabels={$0$, $\frac{3-\sqrt{5}}{4}$, $\frac{1}{3}$, $\frac{1}{2}$},
     ytick = {-0.5,1},
     scale=0.9, restrict y to domain=-1:1]
\addplot [dashed, red,thick, samples=300, smooth,domain=0:0.5] {1+x*ln(x)/ln(2)+(1-x)*ln(1-x)/ln(2)};
\addlegendentry{Capacity for $\bsc^{fb}(p)$}
\addplot [black, thick, samples=300, smooth,domain=0:1/(3+sqrt(5))] {1+x*ln(x)/ln(2)+(1-x)*ln(1-x)/ln(2)};
\addlegendentry{Capacity for $\bsc^{fb}_{adv}(p)$}
\addplot [black,thick, samples=300, smooth,domain=1/(3+sqrt(5)):0.5] {(1-3*x)*ln((1+sqrt(5))/2)/ln(2)};
\node at (axis cs:1/3,-0.5) {$\frac{1}{3}$};
\end{axis}
\end{tikzpicture}
\caption{Capacity for BSC with feedback and adversarial BSC with feedback.}
\end{figure}
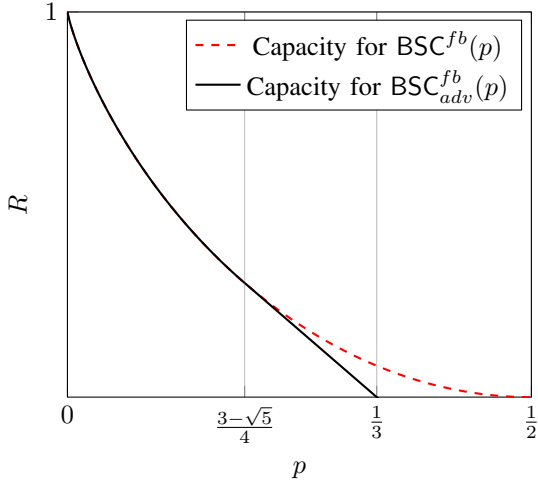

We say that a code $\mathcal{C}$ for the $\bsc_{adv}^{fb}(f)$ is \emph{admissible} if $\mathcal{C}$ can correct any error pattern with error fraction at most $f$. We say that a sequence of codes $\{\mathcal{C}_{N,R}\}_N$ for $\bsc^{fb}(p)$ is \emph{admissible} if the error probability $P_e(\mathcal{C}_{N,R})$ tends to $0$ as $N$ goes to infinity.

\subsection{Markov Chains}
\begin{definition}
A discrete stochastic process $\{X_i\}$ is said to be an $(\ell-1)$-th order
Markov chain if for any $i$,
\begin{align*}
    &\Pr[X_{i}=x_i|X_{1}=x_1,\dots,X_{i-1}=x_{i-1}]\\
    =&\Pr[X_{i}=x_i|X_{i-\ell+1}=x_{i-\ell+1},\dots,X_{i-1}=x_{i-1}],
\end{align*}
for all $x_1,\dots,x_i\in\mathcal{X}$. 

\end{definition}

\subsection{Rubber Method}
Here we briefly present the rubber method for $\bsc^{fb}_{adv}(f)$ \cite{4036418}. Let $\mathcal{A}^{N'}_{\ell}$ denote the set of binary sequences of length $N'$ with no $\ell$ consecutive zeros. Such sequences are called \emph{skeleton sequences}. The sender chooses a skeleton sequence $x^{N'}\in\mathcal{A}^{N'}_{\ell}$ and the decoder's goal is to recover that sequence correctly. The idea is that the encoder can use $\ell$ consecutive zeros to signal an error. Specifically, we have

\begin{itemize}
    \item Decoding $g_R$: the decoder maintains a stack of received bits, which begins empty. Whenever the decoder receives a bit, it inserts the received bit onto the stack and checks if there are consecutive $\ell$ zeros in the stack. If yes, it removes these $\ell$ zeros as well as the bit before these consecutive $\ell$ zeros from the stack. Finally, it truncates the output to $N'$ bits.
    
    \item Encoding $f_R$: if the decoder's current stack is a prefix of $x^{N'}$, then send the next bit in $x^{N'}$. Otherwise send a $0$. If $x^{N'}$ has been sent in its entirety, then send $1$ for all remaining time steps.
\end{itemize}

\begin{center}
\begin{figure}
    \includegraphics[scale = 0.35]{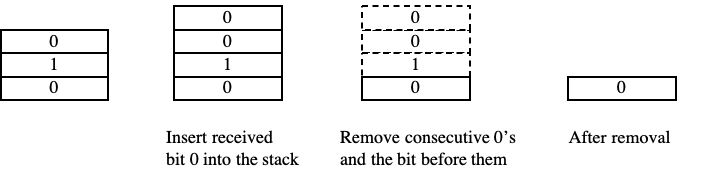}
    \caption{The decoder's stack with $\ell=2$.}
    \label{fig:stack}
    \end{figure}
\end{center}

\begin{proposition}
\label{pro:rubber}
For skeleton sequence set $\mathcal{A}_{\ell}^{N'}$ and block length $N$, a code constructed using the rubber method is admissible for $\bsc_{adv}^{fb}(f)$ if
\begin{equation}
N' + (\ell + 1) fN \le N.
\end{equation}
\end{proposition}
\begin{proof}
See Section 2.2 of \cite{ahlswedenon}.
\end{proof}

\begin{example} 
Suppose the encoder chooses $x = 011010\in\mathcal{A}_2^6$ and the maximum fraction of adversarial errors is $f=1/3$. 

Suppose the first three bits the decoder receives are $010$, which is not a prefix of $x$. The encoder then sends $0$ and suppose decoder sees $0100$. The decoder then erases the last three bits (the consecutive zeros and the one before them) and its stack becomes $0$. This is now a prefix of $x$ and the encoder would thus resend the second bit in $x$, which is $1$. See Figure~\ref{fig:stack}.
\end{example}

\subsection{Shannon–Fano–Elias Code and Arithmetic Coding}

The Shannon–Fano–Elias code compresses a source sequence with known distribution to near-optimal length. It uses the cumulative distribution function $F(x)$ to allot codewords. For a random variable $X\in\{1,2,\dots,M\}$ with distribution $p$, the codeword is $\lfloor\Bar{F}(x)\rfloor_{l(x)}$ where 
\[\Bar{F}(x)=\sum_{a<x}p(a)+\frac{1}{2}p(x),\]
lies between $F(x)$ and $F(x+1)$ and $l(x)=\lceil\log\frac{1}{p(x)}\rceil+1$. In addition, the Shannon-Fano-Elias code is prefix-free. That is, no codeword is a prefix of any other.

Arithmetic coding is an algorithm for efficiently computing the Shannon–Fano–Elias codeword for sequences given a method for computing the probability of the next symbol given the past (e.g.,~\cite[Ch.~4]{sayood2017introduction}). 

\subsection{Constant Recursive Sequences and the Perron–Frobenius Theorem}

\begin{lemma}[Theorem 2.3.6, \cite{cull2005difference}]
A sequence $A(n)$ is an order-$d$ \emph{constant-recursive sequence} if for all $n\geq d+1$, \[A(n)=c_1A(n-1)+c_2A(n-2)+\dots+c_dA(n-d).\] The $n$-th term $A(n)$ in the sequence must be of the form 
\[A(n)=k_1(n)\lambda_1^n+k_2(n)\lambda_2^n+\dots+k_{d'}(n)\lambda_{d'}^n,\]
where $\lambda_i$ is a root with multiplicity $d_i$ of the polynomial
\[\lambda^d-c_1\lambda^{d-1}-\dots-c_d,\]
and $k_i(n)$ is a polynomial with degree $d_i-1$.
\end{lemma}

\begin{definition}[(8.3.16), \cite{meyer2000matrix}]
A matrix $M$ is a \emph{positive (non-negative)} matrix if every entry of $M$ is positive (non-negative). 

A non-negative square matrix $M$ is \emph{primitive} if its $k$-th power is positive for some natural number $k$.
\end{definition}

\begin{lemma}[Perron–Frobenius Theorem, Page 674, \cite{meyer2000matrix}]
If $M$ is a primitive matrix, then $M$ has a positive real eigenvalue $\lambda^*$ such that all other eigenvalues $\lambda_i$ have absolute value $|\lambda_i|<|\lambda^*|$. Moreover, $\lambda^*$ is a simple eigenvalue and its corresponding column and row eigenvectors are positive
\label{lem:perron}. 
\end{lemma}

See \cite[Ch.~8]{meyer2000matrix} for further detail about the Perron-Frobenius Theorem.

\section{A Key Markov Chain}
\label{sec:markov}
In this section we show that we can efficiently compute the distribution of a Markov Chain that is uniformly distributed over $\mathcal{A}_{\ell}^N$. 

Recall that $\mathcal{A}^{N}_{\ell}$ denotes the set of binary sequences of length $N$ with no consecutive $\ell$ zeros.
\if0\short{Let $A_{\ell}(N)=|\mathcal{A}^{N}_{\ell}|$ and let $A_{\ell}(z)$ denote the number of allowable sequences in $\mathcal{A}^{N}_{\ell}$ that begin with $z$ for any binary sequence $z$. }\fi 

\begin{lemma}
\label{lem: entropyrate}
Let $\lambda^*_{\ell}$ be the unique real solution that lies in $(1,2)$ of 

\begin{equation}
\label{eqn:character}
\lambda^{\ell} = \lambda^{\ell-1}+\lambda^{\ell-2}+\dots+1.  
\end{equation}
Then $\lim_{N\rightarrow \infty}\frac{|\mathcal{A}_{\ell}^N|}{\lambda^{*^N}_{\ell}}$ exists and is positive and finite.
\end{lemma}

\if1\short
{
\begin{proof}
We first compute the cardinality of $\mathcal{A}^N_{\ell}$. Let $A_{\ell}(N)$ denote $|\mathcal{A}^N_{\ell}|$. 

Consider all allowable sequences in $\mathcal{A}^N_{\ell}$. The number of sequences in $\mathcal{A}^N_{\ell}$ that begin with $1$ is $A_{\ell}(N-1)$. The number of sequences in $\mathcal{A}^N_{\ell}$ that begin with $01$ is $A_{\ell}(N-2)$, and so on. Continuing recursively we have that
\[A_{\ell}(N) = A_{\ell}(N-1)+A_{\ell}(N-2)+\dots+A_{\ell}(N-\ell).\]


Let $\lambda_1,\dots,\lambda_{\ell'}$ be the roots of equation (\ref{eqn:character}), where $\lambda_i$ has multiplicity $d_i$. Note that $\lambda^{\ell} = \lambda^{\ell-1}+\lambda^{\ell-2}+\dots+1$ is also the characteristic polynomial for the following $\ell\times \ell$ non-negative matrix:
\begin{equation*}
    M_{\ell}^N=
    \begin{pmatrix}
    0 & 1 & 0 &\dots &0\\
    0 & 0 & 1 &\dots &0\\
    \dots & \dots & \dots &\dots &\dots\\
    0 & 0 & 0 &\dots &1\\
    1 & 1 & 1 &\dots &1
    \end{pmatrix}.
\end{equation*}
Therefore $\lambda_1,\dots,\lambda_{\ell'}$ are also the eigenvalues of $M_{\ell}^N$. It is easy to see that equation (\ref{eqn:character}) has exactly one positive real root that lies inside $(1,2)$ and no real root in $[2,+\infty)$. Without loss of generality, we assume that $\lambda_1$ is this root. Moreover, $M_{\ell}^N$ is primitive since $(M_{\ell}^N)^{\ell}$ is a positive matrix. According to Perron–Frobenius theorem, $\lambda_1$ is a simple root with multiplicity $1$ of equation (\ref{eqn:character}) and  $|\lambda_i|<|\lambda_1|$ for $i=2,\dots,\ell'$. Therefore,
\begin{equation}
\label{eqn:AlN_size}
    A_{\ell}(N)=k_1\lambda_1^N+k_2(N)\lambda_2^N+\dots+k_{\ell'}(N)\lambda_{\ell'}^N,
\end{equation}
where $k_i(\cdot)$ is a polynomial with degree $d_i-1$. Since $\lambda_1$ is a simple dominating root and its corresponding column and row eigenvectors are positive, according to Theorem 2.4.2 in \cite{cull2005difference} and Lemma~\ref{lem:perron}, \[\lim_{N\rightarrow\infty}\frac{|A_{\ell}^N|}{\lambda_1^N}=k_1>0\]

\end{proof}
}
\fi

Note that Lemma \ref{lem: entropyrate} implies that
\[\lim_{N\rightarrow\infty}\frac{1}{N}\log|\mathcal{A}_{\ell}^N|=\log\lambda^*_{\ell}.\]
\begin{lemma}
\label{lem:mcGen}
The stochastic process that is uniformly distributed over $\mathcal{A}_{\ell}^N$ is an $(\ell-1)$-th order Markov Chain.
\end{lemma}

\if1\short{
\begin{proof}
Let $z$ be any binary sequence. We abuse the notation slightly by defining $A_{\ell}(z)$ to be the number of allowable sequences in $A_N^{\ell}$ that begin with $z$. 

Suppose $\{X_i\}_{i=1}^N$ is a stochastic process that is uniformly distributed over $\mathcal{A}_{\ell}^N$. Then we have

\begin{itemize}
    \item $\Pr[X_1=1] = \frac{\text{number of sequences begin with }1}{|\mathcal{A}_{\ell}^N|}= \frac{A_{\ell}(N-1)}{A_{\ell}(N)}$;
    \item $\Pr[X_1=0] = 1-\frac{A_{\ell}(N-1)}{A_{\ell}(N)}$;
    \item For $i\geq 2$, for any $z\in\{0,1\}^{i-1}$,
    \begin{align*}
        \Pr[X_{i}=1|X_1,\dots,X_{i-1}&=z]=\frac{A_{\ell}(N-i)}{A_{\ell}(z)},\\
        \Pr[X_{i}=0|X_1,\dots,X_{i-1}&=z]=1-\frac{A_{\ell}(N-i)}{A_{\ell}(z)}.
    \end{align*}
\end{itemize}

To see that $\{X_i\}$ is an $(\ell-1)$-th order Markov Chain, we only need to show that for any $i\geq \ell$, $z\in\{0,1\}^{i-1}$
\begin{align*}
    &\Pr[X_i=1|X_1,\dots,X_{i-1}=z] \\
    =&\Pr[X_i=1|X_{i-\ell+1},\dots,X_{i-1}=z[i-\ell+1,i-1]].
\end{align*}

Fix a $z\in\{0,1\}^{i-1}$ for $i\geq \ell$. Suppose $z$ ends with $\alpha$ zeros. Then $0\leq \alpha\leq \ell-1$ since the sequence is in $\mathcal{A}_{\ell}^N$. We have that
\begin{equation}
\label{eqn:Alz}
A_{\ell}(z) = A_{\ell}(N-i+\alpha+1)-\sum_{k=0}^{\alpha-1} A_{\ell}(N-i+k+1).
\end{equation}

When $\alpha=0$, equation (\ref{eqn:Alz}) becomes $A_{\ell}(z) = A_{\ell}(N-i+\alpha+1)$. This indicates that for any $z$ and $z'$ that have the same last $\ell-1$ bits, $A_{\ell}(z)=A_{\ell}(z')$ and that $\Pr[X_i=1|X_1,\dots,X_{i-1}=z]=\Pr[X_i=1|X_1,\dots,X_{i-1}=z']$. 

Then for any $x\in\{0,1\}^{\ell-1}$ and any $x'\in\{0,1\}^{i-\ell}$, we have
\begin{align*}
    \Pr[X_{i}=1&|X_{i-\ell+1},\dots,X_{i-1}=x]\\
    =\sum_{x''\in\{0,1\}^{i-\ell}}&\left\{\Pr[X_i=1|X_1,\dots,X_{i-1}=x''\|x]\right.\\
    &\left.\cdot\Pr[X^{i-\ell}=x''|X_{i-\ell+1},\dots,X_{i-1}=x]\right\}\\
    =\sum_{x''\in\{0,1\}^{i-\ell}}&\left\{\Pr[X_i=1|X_1,\dots,X_{i-1}=x'\|x]\right.\\
    &\left.\cdot\Pr[X^{i-\ell}=x''|X_{i-\ell+1},\dots,X_{i-1}=x]\right\}\\
    =\Pr[X_i=1&|X_1,\dots,X_{i-1}=x'\|x],
\end{align*}
where the second equation comes from the fact that for any two sequences with the same last $\ell-1$ bits, we have 
\begin{multline*}
    \Pr[X_i=1|X_1,\dots,X_{i-1}=x''\|x] \\ =\Pr[X_i=1|X_1,\dots,X_{i-1}=x'\|x].
\end{multline*}
%
\end{proof}
}\fi

The proof shows that to compute the probability of the next symbol
in the string given the past, we only need to compute 
 $A_{\ell}(N)$ for various values of $N$. This can be computed using $A_{\ell}(N)=k_1\lambda_1^N+k_2(N)\lambda_2^N+\dots+k_{\ell'}(N)\lambda_{\ell'}^N$ where $\lambda_i$ are the roots of equation (\ref{eqn:character}) and $c_1,k_1(N),\dots,k_{\ell'}(N)$ can be determined by the initial conditions $A_{\ell}(1)=2,\dots,A_{\ell}(\ell-1)=2^{\ell-1}, A_{\ell}(\ell)=2^{\ell}-1$. 

Note that when $N$ is large, $A_{\ell}(N)$ is well-approximated as $A_{\ell}(N)\approx k_1\lambda^{*^N}_{\ell}$. Under this approximation the Markov Chain becomes time-invariant.

\begin{example}
Consider the case $\ell = 2$. That is, we  forbid two consecutive zeros in
the skeleton sequence. Then the characteristic polynomial is $\lambda^2-\lambda-1=0$. The two roots are $\lambda_1=\frac{1+\sqrt{5}}{2}$ and $\lambda_2=\frac{1-\sqrt{5}}{2}$ respectively. The initial condition is $A_{\ell}(1)=2, A_{\ell}(2)=3$.
Therefore $A_{\ell}(N)=k_1\lambda_1^N+k_2\lambda_2^N$ where $k_1=\frac{3+\sqrt{5}}{2\sqrt{5}}$, $k_2=\frac{\sqrt{5}-3}{2\sqrt{5}}$. 
See also \cite[Ex.~4.7]{cover2012elements}

\end{example}

\section{A Practical Coding Scheme}
\label{sec:coding}
In this section we combine arithmetic coding and the rubber method to give an efficient feedback code for $\bsc^{fb}_{adv}(f)$ and $\bsc^{fb}(p)$. 
First we describe a modified version of arithmetic coding that will be used in our scheme. 
Consider the following pair of algorithms $({\sf Decom}_{\ell}, {\sf Com}_{\ell})$: 

\begin{mdframed}
\begin{algorithm}
\label{alg:compress}
$({\sf Decom}_{\ell}, {\sf Com}_{\ell})$
\end{algorithm}

Let $L=\lceil\log|\mathcal{A}_{\ell}^N|\rceil$. Let $\{X_i\}_{i=1}^N$ be a stochastic process that is uniformly distributed over $\mathcal{A}_{\ell}^N$. Let $(A_C,A_D)$ where $A_C:\mathcal{A}_{\ell}^N\mapsto\{0,1\}^{L+1}$ and $A_D:\{0,1\}^{L+1}\mapsto \mathcal{A}_{\ell}^N\cup\{\bot\}$
be the compression and decompression algorithms for arithmetic coding applied to $\{X_i\}_{i=1}^N$, where the decompressor outputs $\bot$ if its input
is not a valid codeword. Let $L'$ be any integer such that $L'\leq L-3$.

\noindent\emph{${\sf Decom}_{\ell}(m): \{0,1\}^{L'}\mapsto \mathcal{A}_{\ell}^N$}
\begin{enumerate}
    \item Run the decompress algorithm $A_D(m\|m')$ for all possible $m'\in\{0,1\}^{L+1-L'}$. Let the first non-$\bot$ output be $A_D(m\|m')=x^N$. If there's no such $x^N$, set $x^{N}$ to be a random sequence in $\mathcal{A}_{\ell}^{N}$.
    \item Output $x^N$.
\end{enumerate}

\noindent\emph{${\sf Com}_{\ell}(x^N):\mathcal{A}_{\ell}^N\mapsto\{0,1\}^{L'}$}:
\begin{enumerate}
    \item Output $\lfloor A_C(x^N)\rfloor_{L'}$.
\end{enumerate}
\end{mdframed}
\begin{lemma}
\label{lem:compress}
The pair of algorithms $({\sf Decom}_{\ell}, {\sf Com}_{\ell})$ described in Algorithm \ref{alg:compress} satisfies 
 \[{\sf Com}_{\ell}({\sf Decom}_{\ell}(m))=m, \forall m\in\{0,1\}^{L'}.\]

\end{lemma}

\if1\short{
\begin{proof}
Suppose that all sequences in $\mathcal{A}_{\ell}^N$ are lexicographically sorted and $x^N+1$ is the sequence following $x^N$. Note that for some binary sequences of length $L+1$, $A_D$ might output $\bot$ if the binary sequence is not a Shannon-Fano-Elias codeword for any $x^N\in\mathcal{A}_{\ell}^{N}$.


As long as there exists an $m'$ such that $A_D(m\|m')\neq \bot$, ${\sf Com}({\sf Decom}(m))=m$ due to the correctness of arithmetic coding. Therefore we only need to show that for any $m\in\{0,1\}^{L'}$, there exists $m'\in\{0,1\}^{L+1-L'}$ such that $A_D(m\|m')\neq \bot$.

We will prove that for any $m\in\{0,1\}^{L'}$, there must exist an $x^N\in\mathcal{A}_{\ell}^N$ such that $m$ is a prefix of $A_C(x^N)$. To see this, let each sequence $m$ in $\{0,1\}^{L'}$ represent an interval of length $\frac{1}{2^{L'}}$ in $[0,1]$ such that all of the real numbers inside the interval represented by $m$ have prefix $m$. Note that $A_C(x^N)$ falls between $F(x^N)$ and $F(x^N+1)$, where $F(\cdot)$ is the cumulative distribution function of $\{X_i\}_{i=1}^N$. As $X^N$ is uniformly distributed over $\mathcal{A}_{\ell}^N$, for any $x^N$, $F(x^N+1)-F(x^N) = \frac{1}{|\mathcal{A}_{\ell}^N|}\leq\frac{1}{2^{L'+2}}$. Therefore, for any $m$, the 
interval represented by $m$ with length $\frac{1}{2^{L'}}$ must contain both $F(x^N)$ and $F(x^N+1)$ for at least one $x^N        $. This indicates that $A_C(x^N)\in(F(x^N),F(x^N+1))$ must fall inside the interval represented by $m$. That is, $m$ must be a prefix of $A_C(x^N)$.

\end{proof}
}\fi

Now we describe the construction of our overall scheme:

\begin{mdframed}
\begin{construction}
\label{cst:coding}
The encoding and decoding of $\mathcal{C}_{\ell,N,R}$ are as follows:
\end{construction}

\noindent\emph{Encoding}: 
\begin{itemize}
    \item Let $m^{NR}$ be a message source of length $NR$. Find the minimum natural number $N'$ such that $\lceil\log|\mathcal{A}_{\ell}^{N'}|\rceil\geq NR+3$.
    \item Run ${\sf Decom}_{\ell}(m)$ and denote the output as $x^{N'}$. Let $x^{N'}$ be the skeleton sequence and send it through the feedback channel using the rubber method.
\end{itemize}
\emph{Decoding}: 
\begin{itemize}
    \item Let $y^N$ be the sequence received from the feedback channel. Run the decoding algorithm of the rubber method on $y$ to get $\widetilde{x}^{N'}$. If $\widetilde{x}^{N'}\notin \mathcal{A}_{\ell}^N$, set $\widetilde{x}^{N'}$ to be a random skeleton sequences in $\mathcal{A}_{\ell}^{N'}$.
    \item Otherwise, output $m'={\sf Com}_{\ell}(\widetilde{x}^{N'})$.
\end{itemize}

\end{mdframed}

\begin{proposition}
\label{thm:coding}
The code $\mathcal{C}_{\ell,N,R}$ in Construction \ref{cst:coding} is admissible for the ${\sf BSC}^{fb}_{adv}(f)$ if $N'\leq (1-(\ell+1)f)N$.
\end{proposition}

\if1\short{
\begin{proof}
Follows directly from Proposition \ref{pro:rubber} and Lemma \ref{lem:compress}.
\end{proof}
}\fi

Note that in the first step of encoding, we can find $N'$ simply by computing $\mathcal{A}_{\ell}^{\widetilde{N}}$ for $\widetilde{N}=NR+3,\dots,2NR+6$ since $2^{\frac{\widetilde{N}}{2}}\leq|\mathcal{A}_{\ell}^{\widetilde{N}}|\leq 2^{\widetilde{N}}$. See Lemma \ref{lem:AlN_lb} in Appendix.



We further note that the above coding scheme also works for stochastic feedback BSC channel $\bsc^{fb}(p)$:

\begin{proposition}
\label{cor:coding}
The sequence of codes $\{\mathcal{C}_{\ell,N,R}\}_N$, each of which is constructed as in Construction \ref{cst:coding}, is admissible for the ${\sf BSC}^{fb}(p)$ 
if $R<R_{\ell}(p) = (1-(\ell+1)p)\log\lambda^*_{\ell}$.
\end{proposition}

\if1\short{
\begin{proof}

The fraction of errors that can be corrected by $\mathcal{C}_{\ell,N,R}$ is 
\[f_N=\frac{1}{\ell+1}\left(1-\frac{N'}{N}\right).\]
When $N$ tends to infinity,
\[\lim_{N\rightarrow\infty}f_N=f^*=\frac{1}{\ell+1}\left(1-\frac{R}{\log\lambda^*_{\ell}}\right).\] according to Lemma \ref{lem: entropyrate}.

If the fraction of errors is less than $f_N$, then $\mathcal{C}_{\ell,N,R}$ can decode correctly. Let $E_i$ be the indicator of whether the $i$-th transmitted bit is flipped. 
The error probability of $\mathcal{C}_{\ell,N,R}$ is thus \[P_e(\mathcal{C}_{\ell,N,R})=\Pr\left[\frac{1}{N}\sum_{i=1}^NE_i\geq f_N\right].\]
The result then follows by the law of large numbers.
\end{proof}
}\fi

\section{Main Results}
\label{sec:optimal}
We now show that, for certain parameters, our codes achieve the capacity and the optimal error-exponent, second-order rate, and moderate deviations constant for certain parameters.
\subsection{Capacity}

\begin{theorem}
\label{thm:main}
For any integer $\ell\geq 2$, $R_{\ell}(p)$ is tangent to $C(\bsc^{fb}(p))$. For $p_{\ell}=\frac{1}{1+2^{(\ell+1)\log\lambda^*_{\ell}}}$, \[R_{\ell}(p_{\ell}) = C(\bsc^{fb}(p_{\ell})).\]
That is, for any $\epsilon>0$, the sequence of codes $\{\mathcal{C}_{\ell,N,R}\}_N$ as constructed in Construction \ref{cst:coding} is admissible for $\bsc^{fb}(p_{\ell})$ with $R=C(\bsc^{fb}(p_{\ell}))-\epsilon$.
\end{theorem}

\if1\short{
\begin{proof}

Note that according to Theorem 2 of \cite{ahlswedenon}, $R_{\ell}(p)$ is tangent to $C(\bsc^{fb}(p))$. Moreover according to Section 3.6 of \cite{berlekamp1968block}, when $p=p_{\ell}$, $R_{\ell}(p_{\ell})=C(\bsc^{fb}(p_{\ell}))$. The result then follows from Proposition \ref{cor:coding}.

\end{proof}
}\fi

We call $p_{\ell}$ the \emph{tangent points} and $R^*_{\ell}=R_{\ell}(p_{\ell})$ the \emph{tangent rates}. The tangent points $p_{\ell}$, tangent rates $R^*_{\ell}$, and $\log\lambda^*_{\ell}$ values for different $\ell$ are listed in Table \ref{tab:tangent}.
\begin{table}
    \centering
    \begin{tabular}{c|c|c|c}
         $\ell$ & $\log\lambda^*_{\ell}$ & $p_{\ell}$ & $R^*_{\ell}$ \\
         \hline
          2 & 0.6942 & 0.1910 & 0.2965\\
          3 & 0.8791 & 0.0804 & 0.5965\\
          4 & 0.9468 & 0.0362 & 0.7754
    \end{tabular}
    \caption{Numerical results of $\log\lambda^*_{\ell}$, tangent points $p_{\ell}$ and tangent rates $R^*_{\ell}$ }
    \label{tab:tangent}
\end{table}

The function $R_{\ell}(p)$ for different $\ell$ is plotted in Figure \ref{fig:rubber}.
That the rubber method would achieve the capacity of the $\bsc^{fb}(p_{\ell})$ is 
implicit in \cite{ahlswedenon}. We consider three more-refined performance
measures. 
\begin{figure}
\centering
\begin{tikzpicture}[trim axis left]
\begin{axis}[xmin=0, xmax=1/2, ymin=0, ymax=1, 
    xlabel={$p$}, ylabel={$R$},
     xtick = {0,1/5,1/4,1/3,1/2},xticklabels={$0$, $\frac{1}{5}$, $\frac{1}{4}$,$\frac{1}{3}$, $\frac{1}{2}$}]
     ytick = {0,1}, scale=0.7, restrict y to domain=-1:1]
\addplot [black, thick, samples=300, smooth,domain=0:0.5] {1+x*ln(x)/ln(2)+(1-x)*ln(1-x)/ln(2)};
\addlegendentry{$C(\bsc^{fb}(p))$}
\addplot [dotted, thick, samples=300, smooth,domain=0:1/3] {(1-3*x)*ln((1+sqrt(5))/2)/ln(2)};
\addlegendentry{$R_2(p)$}
\addplot [dashed,thick, samples=300, smooth,domain=0:1/4] {(1-4*x)*0.8791};
\addlegendentry{$R_3(p)$}
\addplot [dashdotted,thick, samples=300, smooth,domain=0:1/5] {(1-5*x)*0.9468};
\addlegendentry{$R_4(p)$}
\end{axis}
\end{tikzpicture}
\caption{$R_{\ell}(p)$ for different $\ell$.}
\label{fig:rubber}
\end{figure}
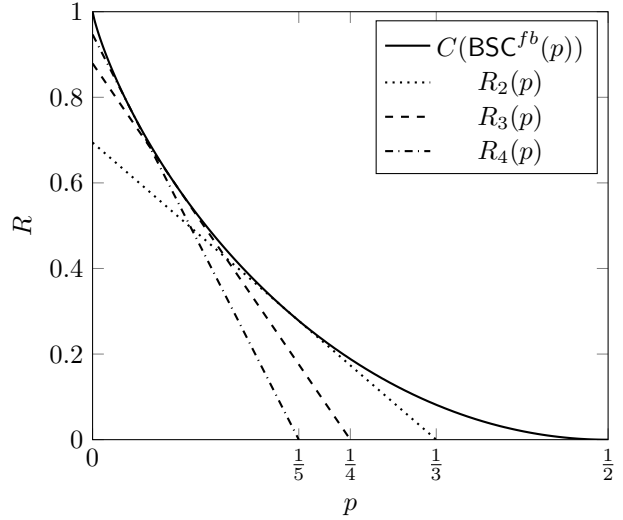

\subsection{Error-exponent}
\begin{lemma}[Sphere-packing bound with pre-factor \cite{Elias:London,altuug2020exact}]
Let $\{\mathcal{C}_{N,R}\}_N$ be a sequence of codes for the ${\sf BSC}^{fb}(p)$, each with rate $R<C({\sf BSC}^{fb}(p))$. Let $q\in(0,\frac{1}{2})$ s.t. $R=1-h(q)$. Let $E_{sp}(R)=D(B(q)\|B(p))$ and $E'_{sp}(R)$ be the slope of the error exponent at $R$. 
Then the error probability $P_e(\mathcal{C}_{N,R})$ satisfies
\[P_e(\mathcal{C}_{N,R})\geq \frac{K_1}{N^{\frac{1}{2}(1+|E'_{sp}(R)|)}}e^{-NE_{sp}(R)},\]
where $K_1$ is a positive constant depending on $R$.
\label{thm:opt_error_exp}
\end{lemma}

\begin{theorem}
\label{thm:error_exp}
For any fixed $\ell\geq 2$, consider the sequence of codes $\{\mathcal{C}_{\ell,N,R^*_{\ell}}\}_N$ at the tangent rate $R^*_{\ell}$. That is, $R^*_{\ell} = R_{\ell}(p_{\ell})=1-h(p_{\ell})$. Then for the $\bsc^{fb}(p)$ with $p<p_{\ell}$, $\{\mathcal{C}_{\ell,N,R^*_{\ell}}\}_N$ at rate $R^*_{\ell}$ achieves optimal error exponent
\[P_e(\mathcal{C}_{\ell,N,R^*_{\ell}})\leq O\left(\frac{1}{\sqrt{N}}\right)e^{-N\cdot E_{sp}(R)}.\]
In particular,
\[\lim_{N\rightarrow\infty}-\frac{1}{N}\log P_e(\mathcal{C}_{\ell,N,R^*_{\ell}})= E_{sp}(R).\]
\end{theorem}

\begin{remark}
The ``pre-factor'' order achieved by our scheme is $O(\frac{1}{\sqrt{N}})$, which 
is worse than the optimal order of $O(\frac{1}{N^{\frac{1}{2}(1+|E'_{sp}(R)|)}})$ in Theorem~\ref{thm:opt_error_exp}. Interestingly, for the binary erasure
channel (BEC), both with and without feedback, the optimal pre-factor is $O(\frac{1}{\sqrt{N}})$~\cite[Theorem~2]{altuug2020exact}. 
Rubber coding attempts to emulate a BEC using the BSC, which might explain
this connection. A similar gap from strict optimality occurs in the second-order
coding rate results to follow. Making the connection between rubber coding
and the BEC more precise is an interesting topic for
future study.
\label{remark:suboptimal}
\end{remark}

\if1\short{
\begin{proof}
Let $R_0=\log\lambda^*_{\ell}$. Let $f_N=\frac{1}{\ell+1}(1-\frac{N'}{N})$ be the fraction of errors $\mathcal{C}_{\ell,N,R^*_{\ell}}$ can correct. Since $\lceil\log|\mathcal{A}_{\ell}^{N'}|\rceil\geq NR_N+3$, and  $\lceil\log|\mathcal{A}_{\ell}^{N'-1}|\rceil< NR_N+3$, we have
\[\frac{N'}{N}\leq \frac{R^*_{\ell}}{\log\lambda^*_{\ell}}+O\left(\frac{1}{N}\right)
.\]
This indicates that 
\[f_N=\frac{1}{\ell+1}\left(1-\frac{N'}{N}\right)=p_{\ell}+O\left(\frac{1}{N}\right).\]

If the number of errors is less than $Nf_N$, $\mathcal{C}_{\ell,N,R^*_{\ell}}$ can correctly decode the message. Define $r_N = \frac{p}{f_N}\frac{1-f_N}{1-p}$. Let $E_i$ be the indicator random variable of whether the $i$-th bit is flipped. 
By Lemma~\ref{lem:tail}, when $N$ is large, the error probability $P_e(\mathcal{C}_{\ell,N,R^*_{\ell}})$ satisfies

\begin{align*}
   &P_e(\mathcal{C}_{\ell,N,R^*_{\ell}})=\Pr\left[\sum_{i=1}^NE_i\geq Nf_N\right]\\
   \leq &\frac{e^{-ND(B(f_N)\|B(p))}}{\sqrt{2\pi f_N(1-f_N)N}}
    \left(a_N+o\left(\frac{1-r_N^{\lfloor(1-f_N)N\rfloor+1}}{1-r_N}\right)\right),
\end{align*}
where
\begin{align*}
    a_N=\frac{1-r_N^{\lfloor(1-f_N)N\rfloor+1}\exp{-(\frac{\lfloor(1-f_N)N\rfloor+1}{2f_N(1-f_N)N})}}{1-r_N\exp{(-\frac{1}{2f_N(1-f_N)N})}}.
\end{align*}
Since $D(B(\cdot)\|B(p))$ is continuous, \[D(B(f_N)\|B(p))=D(B(p_{\ell})\|B(p))+O\left(\frac{1}{N}\right).\]
Therefore,
\begin{align*}
    P_e(\mathcal{C}_{\ell,N,R^*_{\ell}})&\leq O\left(\frac{1}{\sqrt{N}}\right)e^{-N(E_{sp}(R)+O(\frac{1}{N}))}\\
    &=O\left(\frac{1}{\sqrt{N}}\right)e^{-NE_{sp}(R)}.
\end{align*}
\end{proof}
}\fi

\subsection{Second-order Rate}
\begin{lemma}[Second-order coding rate: Theorem 15, \cite{5961844}]
Given a block length $N$ and an $\epsilon$ such that $0<\epsilon<1$, the largest possible rate of a code for the ${\sf BSC}^{fb}(p)$ with error probability less than or equal to $\epsilon$ is 
\[C-\frac{1}{\sqrt{N}}\sqrt{p(1-p)\log^2\frac{1-p}{p}}\Phi^{-1}(1-\epsilon)+\frac{\log N}{2N}+o(1),\]
where $\Phi$ denotes the standard Gaussian distribution.
\end{lemma}

\begin{theorem}
\label{thm:second_order_rate}
For any fixed $\ell\geq 2$, consider the  $\bsc^{fb}(p)$ with cross-over probability $p=p_{\ell}$. Fix $\epsilon\in(0,1)$, and let $R(N,\epsilon)$ denote the largest possible rate $R$ such that $\mathcal{C}_{\ell,N,R(N,\epsilon)}$ has error probability at most $\epsilon$, and let $C$ denote the capacity of the $\bsc^{fb}(p)$. Then for large $N$,

\begin{align*}
    &R(N,\epsilon)\\
    \geq &C-\frac{1}{\sqrt{N}}\sqrt{p(1-p)\log^2\frac{1-p}{p}}\Phi^{-1}(1-\epsilon)-O\left(\frac{1}{N}\right).
\end{align*}

\end{theorem}

\begin{remark}
Note the $\log N/N$ term is ``missing'' from the expansion in Theorem~\ref{thm:second_order_rate}. See Remark~\ref{remark:suboptimal}.
\end{remark}

\if1\short{
\begin{proof}

Let 
\[R_N=C-\frac{1}{\sqrt{N}}\sqrt{p(1-p)\log^2\frac{1-p}{p}}\Phi^{-1}(1-\epsilon)-\frac{c_0}{N},\] where $c_0$ is a positive constant which we will specify later. We now show that for sufficiently large $N$, the error probability of  $\mathcal{C}_{\ell,N,R_N}$ satisfies $P_e(\mathcal{C}_{\ell,N,R_N})< \epsilon$.

Let $e^*_N=\frac{N}{\ell+1}(1-\frac{N'}{N})$ denote the number of errors that $\mathcal{C}_{\ell,N,R_N}$ is capable of correcting. According to our construction, $\lceil\log|\mathcal{A}_{\ell}^{N'}|\rceil\geq NR_N+3$, and $\lceil\log|\mathcal{A}_{\ell}^{N'-1}|\rceil< NR_N+3$, we have
\[\frac{N'}{N}\leq \frac{R_N}{\log\lambda^*_{\ell}}+\frac{c_1}{N}+o\left(\frac{1}{N}\right)
,\]
where $c_1=\frac{3-\frac{1}{2}\log k_1}{\log\lambda^*_{\ell}}$. Therefore
\[e_N^*\geq \frac{N}{\ell+1}\left(1-\frac{R}{\log\lambda^*_{\ell}}\right)-\frac{c_1}{\ell+1}-o(1).\]
Let $E_i$ be the random variable such that $E_i=1$ if the $i$-th bit is flipped. Let $\Psi_N$ be the c.d.f. of the binomial distribution ${\sf Bin}(N,p)$. According to Berry–Esseen theorem (Section 5, \cite{bhattacharya2007basic}), for any $N$, for any $x$,

\[\left|\Psi_N(x\sigma\sqrt{N}+Np)-\Phi(x)\right|\leq \frac{c_2}{\sqrt{N}},\]
where $\Phi$ is the c.d.f. of standard Gaussian and $\sigma = \sqrt{p(1-p)}$, $c_2=\frac{0.56p}{\sigma^3}$. Therefore
\begin{align*}
    &P_e(\mathcal{C}_{\ell,N,R_N})\leq 1-\Psi_N(e_N^*)\\
    \leq&1-\Psi_N\left(\frac{N}{\ell+1}\left(1-\frac{R_N}{\log\lambda^*_{\ell}}\right)-\frac{c_1}{\ell+1}-o(1)\right)\\
    \leq&1-\Phi\left(\frac{1}{\sigma\sqrt{N}}\left(\frac{N}{\ell+1}\left(1-\frac{R_N}{\log\lambda^*_{\ell}}\right)\right.\right.\\
    &-\left.\left.\frac{c_1}{\ell+1}-o(1)-Np\right)\right)+\frac{c_2}{\sqrt{N}}.
\end{align*}
Let $R_0=\log\lambda^*_{\ell}$. Note that
\begin{align*}
    &\frac{1}{\sigma\sqrt{N}}\left(\frac{N}{\ell+1}\left(1-\frac{R_N}{\log\lambda^*_{\ell}}\right)-\frac{c_1}{\ell+1}-o(1)-Np\right)\\
    =&\frac{1}{\sigma\sqrt{N}}\frac{N}{\ell+1}\left[1-(\ell+1)p-\frac{C}{R_0}\right.\\
    &\left.+\frac{\sigma}{\sqrt{N}}(\ell+1)\Phi^{-1}(1-\epsilon)\right]\\
    &+\frac{1}{\sigma\sqrt{N}}\left(\frac{c_0}{(\ell+1)R_0}-\frac{c_1}{(\ell+1)}-o(1)\right)\\
    =&\Phi^{-1}(1-\epsilon)+\frac{1}{\sigma\sqrt{N}}\left(\frac{c_0}{(\ell+1)R_0}-\frac{c_1}{(\ell+1)}-o(1)\right),I'
\end{align*}
where the first equality comes from the fact that when $p=p_{\ell}$,  $\log\frac{1-p}{p}=R_0(\ell+1)$. See Section 3.6 in \cite{berlekamp1968block}. The second equality comes from the fact that $C=(1-(\ell+1)p)R_0$. Therefore 
\begin{align*}
    &P_e(\mathcal{C}_{\ell,N,R_N})\\
    \leq &1-\Phi\left(\Phi^{-1}(1-\epsilon)\right.\\
    &\left.+\frac{1}{\sigma\sqrt{N}}\left(\frac{c_0}{(\ell+1)R_0}-\frac{c_1}{(\ell+1)}-o(1)\right)\right)+\frac{c_2}{\sqrt{N}}\\
    =&1-\left[\Phi(\Phi^{-1}(1-\epsilon))+\frac{\Phi'(\Phi^{-1}(1-\epsilon))}{\sqrt{N}}\left(\frac{c_0}{(\ell+1)R_0}\right.\right.\\
    &\left.\left.-\frac{c_1}{(\ell+1)}\right)\right]+o\left(\frac{1}{\sqrt{N}}\right)+\frac{c_2}{\sqrt{N}}\\
    =&\epsilon-\frac{\Phi'(\Phi^{-1}(1-\epsilon))}{\sqrt{N}}\left(\frac{c_0}{(\ell+1)R_0}-\frac{c_1}{(\ell+1)}\right)\\
    &+\frac{c_2}{\sqrt{N}}+o(\frac{1}{\sqrt{N}}).
\end{align*}

For $N$ large enough, $o(\frac{1}{\sqrt{N}})<\frac{1}{\sqrt{N}}$. By picking 
\[c_0\geq\left( \frac{c_2+1}{\Phi'(\Phi^{-1}(1-\epsilon))}+\frac{c_1}{\ell+1}\right)(\ell+1)R_0,\] we have that $\frac{\Phi'(\Phi^{-1}(1-\epsilon))}{\sqrt{N}}(\frac{c_0}{(\ell+1)R_0}-\frac{c_1}{(\ell+1)}))-\frac{c_2}{\sqrt{N}}-o(\frac{1}{\sqrt{N}})$ is positive eventually, which implies $P_e(\mathcal{C}_{\ell,N,R_N})< \epsilon$.
\end{proof}
}\fi

\subsection{Moderate Deviations}
\begin{lemma}[Moderate deviations, Corollary 1, \cite{7282769}]
For any sequence of real numbers $\epsilon_N$ s.t. $\epsilon_N\rightarrow 0$ as $ N\rightarrow\infty$ and $\epsilon_N\sqrt{N}\rightarrow \infty$ as $ N\rightarrow\infty$, for any sequence of codes $\{\mathcal{C}_{N,R_N}\}_N$ for the $\bsc^{fb}(p)$ such that $R_{N}\geq C(\bsc^{fb}(p))-\epsilon_N$, we have
\[\liminf_{N\rightarrow\infty} \frac{1}{N\epsilon_N^2}\log P_{e}(\mathcal{C}_{N,R_N})\geq-\frac{1}{2p(1-p)\log^2\frac{1-p}{p}}.\]
\end{lemma}

\begin{theorem}
\label{thm:moderate}
Fix any $\ell\geq 2$. Let $C$ be the capacity of the $\bsc^{fb}(p_{\ell})$. 
For any sequence of real numbers $\epsilon_N$ s.t. $\epsilon_N\rightarrow 0$ as $ N\rightarrow\infty$ and $\epsilon_N\sqrt{N}\rightarrow \infty$ as $ N\rightarrow\infty$, consider the sequence of codes $\{\mathcal{C}_{\ell,N,R_{N}}\}_N$ such that $R_{N}= C-\epsilon_N$. Let $P_{e}(\mathcal{C}_{\ell,N,R_{N}})$ denote the average error probability of $\mathcal{C}_{\ell,N,R_{N}}$ over the $\bsc^{fb}(p_{\ell})$. Then 
\[\lim_{N\rightarrow\infty} \frac{1}{N\epsilon_N^2}\log P_{e}(\mathcal{C}_{\ell,N,R_{N}})=-\frac{1}{2p(1-p)\log^2\frac{1-p}{p}}.\]
\end{theorem}
\if1\short{
\begin{proof}
Let $R_0=\log\lambda^*_{\ell}$. Note that $\mathcal{C}_{\ell,N,R_{N}}$ has rate $R_{N}= C-\epsilon_N$. The maximum fraction of errors it can correct is thus 

\begin{align*}
    f_N &= \frac{1}{\ell+1}\left(1-\frac{N'}{N}\right)=\frac{1}{\ell+1}\left(1-\frac{C-\epsilon_N}{R_0}\right)+O\left(\frac{1}{N}\right)\\
    &=p_{\ell}+\frac{\epsilon_N}{(\ell+1)R_0}+O\left(\frac{1}{N}\right).
\end{align*}
Let $E_i$ be the indicator random variable of whether the $i$-th bit is flipped. Then the error probability $P_{e}(\mathcal{C}_{\ell,N,R_{N}})\leq\Pr[\sum_{i=1}^NE_i\geq Nf_N]$ and $P_{e}(\mathcal{C}_{\ell,N,R_{N}})\geq\frac{1}{2}\Pr[\sum_{i=1}^NE_i\geq Nf_N]$. Let $\epsilon_N'=(f_N-p_{\ell})(\ell+1)R_0$. Then

\begin{align*}
    &\lim_{N\rightarrow\infty}\frac{\epsilon_N}{\epsilon_N'}=\lim_{N\rightarrow\infty}\frac{(f_N-p_{\ell}-O(\frac{1}{N}))(\ell+1)R_0}{(f_N-p_{\ell})(\ell+1)R_0}=1,
\end{align*}
where the last step comes from the fact that $\epsilon_N=\Omega(\frac{1}{\sqrt{N}})$, $f_N-p_{\ell}=\Omega(\frac{1}{\sqrt{N}})$. Define $Z_N=\frac{1}{N\epsilon_N'}\sum_{i=1}^N(E_i-p)$. Then we have
\begin{align}
    &\lim_{N\rightarrow\infty} \frac{1}{N\epsilon_N^2}\log P_{e}(\mathcal{C}_{\ell,N,R_{N}})\nonumber\\
    =&\lim_{N\rightarrow\infty} \frac{1}{N\epsilon_N^2}\log \Pr\left[\sum_{i=1}^NE_i\geq Nf_N\right]\nonumber\\
    =&\lim_{N\rightarrow\infty} \frac{1}{N\epsilon_N^2}\log \Pr\left[Z_N\geq\frac{f_N-p_{\ell}}{\epsilon_N'}\right]\nonumber\\
    =&\lim_{N\rightarrow\infty} \frac{1}{N\epsilon_N^2}\log \Pr\left[Z_N\geq\frac{1}{(\ell+1)R_0}\right]\nonumber\\
    =&\lim_{N\rightarrow\infty}\frac{\epsilon_N'^2}{\epsilon_N^2} \frac{1}{N\epsilon_N'^2}\log \Pr\left[Z_N\geq\frac{1}{(\ell+1)R_0})\right]\nonumber\\
    =&-\frac{1}{2p(1-p)\log^2\frac{1-p}{p}}\nonumber.
\end{align}
where the last equation comes from Theorem 3.7.1 in \cite{dembo1998large} and the fact that when $p=p_{\ell}$, $(\ell+1)R_0=\log\frac{1-p}{p}$.

\end{proof}
}\fi
\section{Appendix}

\begin{lemma}
\label{lem:tail}
Let $E_1,\dots,E_N$ be i.i.d. random variables with $E_1\sim B(p)$. Let $f_N$ be a sequence of real numbers converging to $f^*\in(0,1)$ such that $f^*>p$. Then for large $N$,
\begin{align*}
   &\Pr\left[\sum_{i=1}^NE_i\geq Nf_N\right]\\
   \leq &\frac{e^{-ND(B(f_N)\|B(p))}}{\sqrt{2\pi f_N(1-f_N)N}}
    \left(a_N+o\left(\frac{1-r_N^{\lfloor(1-f_N)N\rfloor+1}}{1-r_N}\right)\right),
\end{align*}
where
\begin{align*}
    r_N&=\frac{p}{f_N}\frac{1-f_N}{1-p},\\
    a_N&=\frac{1-r_N^{\lfloor(1-f_N)N\rfloor+1}\exp{-(\frac{\lfloor(1-f_N)N\rfloor+1}{2f_N(1-f_N)N})}}{1-r_N\exp{(-\frac{1}{2f_N(1-f_N)N})}}.
\end{align*}
\end{lemma}

\if1\short{
\begin{proof}
We follow Theorem 2 in \cite{arratia1989tutorial}. For any fixed $N$, let $Y_1,\dots, Y_N$ be i.i.d. random variables with $Y_1\sim B(f_N)$. For any integer $S\in[0,N]$, we have that
\begin{align*}
    &\Pr\left[\sum_{i=1}^N E_i=S\right]=\binom{N}{S}p^S(1-p)^{N-S},\\
    &\Pr\left[\sum_{i=1}^N Y_i=S\right]=\binom{N}{S}f_N^S(1-f_N)^{N-S}.
\end{align*}
Therefore for any integer $j$, for large $N$,
\begin{align*}
    \Pr&\left[\sum_{i=1}^N E_i=\lceil Nf_N\rceil +j\right]\\
    =\Pr&\left[\sum_{i=1}^N Y_i=\lceil Nf_N\rceil +j \right]\left(\frac{p}{f_N}\right)^{\lceil Nf_N\rceil+j}\\
    &\cdot\left(\frac{1-p}{1-f_N}\right)^{\lfloor N(1-f_N)\rfloor-j}\\
    \leq\Pr&\left[\sum_{i=1}^N Y_i=\lceil Nf_N\rceil +j\right]\left(\frac{p}{f_N}\right)^{ Nf_N+j}\\
    &\cdot\left(\frac{1-p}{1-f_N}\right)^{ N(1-f_N)-j}\\
    =\Pr&\left[\sum_{i=1}^N Y_i=\lceil Nf_N\rceil+j\right]e^{-ND(B(f_N)\|B(p))}r_N^j,
\end{align*}
where the inequality comes from the fact that when $N$ is large, $f_N>p$. Then we have
\begin{align*}
    &\Pr\left[\sum_{i=1}^N E_i\geq Nf_N\right]\\
    =&\sum_{j=0}^{\lfloor(1-f_N)N\rfloor}\Pr\left[\sum_{i=1}^N E_i= \lceil Nf_N\rceil+j\right]\\
    \leq&e^{-ND(B(f_N)\|B(p))}\\
    &\cdot\sum_{j=0}^{\lfloor(1-f_N)N\rfloor}\Pr\left[\sum_{i=1}^N Y_i=\lceil Nf_N\rceil+j\right]r_N^j.
\end{align*}
According to the local central limit theorem (see Theorem 2 of \cite{petrov1964local}), for any $j=0,1,\dots,\lfloor(1-f_N)N\rfloor$
\begin{align*}
    &\Pr\left[\sum_{i=1}^N Y_i=\lceil Nf_N\rceil+j\right]\\
    \leq &\frac{1}{\sqrt{2\pi f_N(1-f_N)N}}\exp\left({-\frac{j^2}{2f_N(1-f_N)N}}\right)\\
    &+o\left(\frac{1}{\sqrt{N}}\right)\\
    \leq &\frac{1}{\sqrt{2\pi f_N(1-f_N)N}}\exp\left({-\frac{j}{2f_N(1-f_N)N}}\right)\\
    &+o\left(\frac{1}{\sqrt{N}}\right).
\end{align*}

Plugging back we have,
\begin{align*}
    &\Pr\left[\sum_{i=1}^N E_i\geq Nf_N\right]\\
    \leq&\frac{e^{-ND(B(f_N)\|B(p))}}{\sqrt{2\pi f_N(1-f_N)N}}\\
    &\cdot\sum_{j=0}^{\lfloor(1-f_N)N\rfloor} r_N^j\left[\exp\left({-\frac{j}{2f_N(1-f_N)N}}\right)+o(1)\right]\\
    =&\frac{e^{-ND(B(f_N)\|B(p))}}{\sqrt{2\pi f_N(1-f_N)N}}
    \left(a_N+o\left(\frac{1-r_N^{\lfloor(1-f_N)N+1}\rfloor}{1-r_N}\right)\right).
\end{align*}
\end{proof}
}\fi

\begin{lemma}
\label{lem:AlN_lb}
For any $N$, $2^{\frac{N}{2}}\leq A_{\ell}(N)\leq 2^N$.
\end{lemma}

\if1\short{
\begin{proof}
It follows directly from the definition that $A_{\ell}(N)\leq 2^N$. To see that $2^{\frac{N}{2}}\leq A_{\ell}(N)$, we use induction on $N$.

Note that the initial conditions, $A_{\ell}(1)=2,\dots,A_{\ell}(N-1)=2^{\ell-1}$, $A_{\ell}^N=2^{\ell}-1$ all satisfy the condition. Suppose that $2^{\frac{N}{2}}\leq A_{\ell}(N)$ holds for all $i\leq k$. Then
\begin{align*}
    &A_{\ell}(k)=A_{\ell}(k-1)+\dots+A_{\ell}(k-\ell)\\
    \geq &2^{\frac{k-1}{2}}+\dots+2^{\frac{k-\ell}{2}}\\
    =&\frac{\sqrt{2}^{k-\ell}-\sqrt{2}^{k}}{1-\sqrt{2}}\\
    \geq &2^{\frac{k}{2}},
\end{align*}
where the last inequality comes from the fact that $\sqrt{2}^{\ell+1}\leq 2\sqrt{2}^{\ell}-1$.
Therefore $2^{\frac{N}{2}}\leq A_{\ell}(N)\leq 2^N$.
\end{proof}
}\fi

\section*{Acknowledgment}
This research was supported by the US National Science Foundation under grant  CCF-1956192 
and the US Army Research Office under grant W911NF-18-1-0426.
\newpage
\bibliographystyle{IEEEtran}
\bibliography{reference}
\end{document}